\documentclass[12pt]{article}
\usepackage{fullpage}
\usepackage{mathrsfs}
\usepackage{float}
\usepackage{amsmath, amsfonts, amssymb, amsthm, amsbsy, amscd, bm, bbm}
\usepackage{array}
\usepackage{booktabs}
\usepackage{graphicx}
\usepackage[small,bf]{caption}
\usepackage{subcaption}
\usepackage{color}
\usepackage{ifthen}
\usepackage{xspace}

\usepackage[linesnumbered, ruled]{algorithm2e} 

\SetAlFnt{\small}
\SetAlCapFnt{\small}
\SetAlCapNameFnt{\small}
\SetAlCapHSkip{0pt}
\IncMargin{-\parindent}

\usepackage[colorlinks,citecolor={black},urlcolor={black},linkcolor={black}]{hyperref}
\usepackage{url}
\usepackage{tocbibind}

\newtheorem{theorem}{Theorem}[section]
\newtheorem{corollary}[theorem]{Corollary}

\newtheorem{definition}[theorem]{Definition}
\newtheorem{lemma}[theorem]{Lemma}
\newtheorem{remark}{Remark}[section]

\newcommand{\1}{\mathbbm{1}}

\DeclareMathOperator{\Prob}{Pr}
\DeclareMathOperator{\E}{\mathbb{E}}
\def \Rev  {{\sf Rev}}
\def \Reg  {{\sf Reg}}

\title{Multi-armed Bandit Problems with Strategic Arms}

\author{
Mark Braverman \thanks{Department of Computer Science, Princeton University, email: mbraverm@cs.princeton.edu. Research supported in part by an NSF CAREER award (CCF-1149888), NSF CCF-1215990, NSF CCF-1525342, NSF CCF-1412958, a Packard Fellowship in Science and Engineering, and the Simons Collaboration on Algorithms and Geometry.}
\and
Jieming Mao  \thanks{Department of Computer Science, Princeton University, email: jiemingm@cs.princeton.edu.}
\and
Jon Schneider \thanks{Department of Computer Science, Princeton University, email: js44@cs.princeton.edu}
\and 
S. Matthew Weinberg \thanks{Department of Computer Science, Princeton University, email: smweinberg@princeton.edu. Research completed in part while the author was a Research Fellow at the Simons Institute for the Theory of Computing. }
}
\begin{document}
\maketitle

\begin{abstract}
We study a strategic version of the multi-armed bandit problem, where each arm is an individual strategic agent and we, the principal, pull one arm each round. When pulled, the arm receives some private reward $v_a$ and can choose an amount $x_a$ to pass on to the principal (keeping $v_a-x_a$ for itself). All non-pulled arms get reward $0$. Each strategic arm tries to maximize its own utility over the course of $T$ rounds. Our goal is to design an algorithm for the principal incentivizing these arms to pass on as much of their private rewards as possible.

When private rewards are stochastically drawn each round ($v_a^t \leftarrow D_a$), we show that:
\begin{itemize}
\item Algorithms that perform well in the classic adversarial multi-armed bandit setting necessarily perform poorly: For all algorithms that guarantee low regret in an adversarial setting, there exist distributions $D_1,\ldots,D_k$ and an approximate Nash equilibrium for the arms where the principal receives reward $o(T)$. 
\item Still, there exists an algorithm for the principal that induces a game among the arms where each arm has a dominant strategy. When each arm plays its dominant strategy, the principal sees expected reward $\mu'T - o(T)$, where $\mu'$ is the second-largest of the means $\E[D_{a}]$. This algorithm maintains its guarantee if the arms are non-strategic ($x_a = v_a$), and also if there is a mix of strategic and non-strategic arms.

\end{itemize}

\end{abstract}

\section{Introduction}

Classically, algorithms for problems in machine learning assume that their inputs are drawn either stochastically from some fixed distribution or chosen adversarially. In many contexts, these assumptions do a fine job of characterizing the possible behavior of problem inputs. Increasingly, however, these algorithms are being applied to contexts (ad auctions, search engine optimization, credit scoring, etc.) where the quantities being learned are controlled by rational agents with external incentives. To this end, it is important to understand how these algorithms behave in \textit{strategic} settings.

The multi-armed bandit problem is a fundamental decision problem in machine learning that models the trade-off between exploration and exploitation, and is used extensively as a building block in other machine learning algorithms (e.g. reinforcement learning). A learner (who we refer to as the \textit{principal}) is a sequential decision maker who at each time step $t$, must decide which of $k$ arms to `pull'. Pulling this arm bestows a reward (either adversarially or stochastically generated) to the principal, and the principal would like to maximize his overall reward. Known algorithms for this problem guarantee that the principal can do approximately as well as the best individual arm. 

In this paper, we consider a strategic model for the multi-armed bandit problem where each arm is an individual strategic agent and each round one arm is pulled by an agent we refer to as the \textit{principal}. Each round, the pulled arm receives a private reward $v \in [0,1]$ and then decides what amount $x$ of this reward gets passed on to the principal (upon which the principal receives utility $x$ and the arm receives utility $v-x$). Each arm therefore has a natural tradeoff between keeping most of its reward for itself and passing on the reward so as to be chosen more frequently. Our goal is to design mechanisms for the principal which simultaneously learn which arms are valuable while also incentivizing these arms to pass on most of their rewards. 

This model captures a variety of dynamic agency problems, where at each time step the principal must choose to employ one of $K$ agents to perform actions on the principal's behalf, where the agent's cost of performing that action is unknown to the principal (for example, hiring one of $K$ contractors to perform some work, or hiring one of $K$ investors with external information to manage some money). In this sense, this model can be thought of as a multi-agent generalization of the principal-agent problem in contract theory (see Section \ref{relatedwork} for references). The model also captures, for instance, the interaction between consumers (as the principal) and many sellers deciding how steep a discount to offer the consumers - higher prices now lead to immediate revenue, but offering better discounts than your competitors will lead to future sales. In all domains, our model aims to capture settings where the principal has little domain-specific or market-specific knowledge, and can really only process the reward they get for pulling an arm and not any external factors that contributed to that reward.

\subsection{Our results}

\subsubsection{Low-regret algorithms are far from strategyproof}

Many algorithms for the multi-armed bandit problem are designed to work in worst-case settings, where an adversary can adaptively decide the value of each arm pull. Here, algorithms such as EXP3 (\cite{AuerCNS03}) guarantee that the principal receives almost as much as if he had only pulled the best arm. Formally, such algorithms guarantee that the principal experiences at most $O(\sqrt{T})$ regret over $T$ rounds compared to any algorithm that only plays a single arm (when the adversary is oblivious).

Given these worst-case guarantees, one might naively expect low-regret algorithms such as EXP3 to also perform well in our strategic variant. It is important to note, however, that single arm strategies perform dismally in this strategic setting; if the principal only ever selects one arm, the arm has no incentive to pass along any surplus to the principal. In fact, we show that the objectives of minimizing adversarial regret and performing well in this strategic variant are fundamentally at odds.

\begin{theorem}[informal restatement of Theorem \ref{thm:advtacit}]\label{thm:introthm1}
Let $M$ be a low-regret algorithm for the classic multi-armed bandit problem with adversarially chosen values. Then there exists an instance of the strategic multi-armed bandit problem and an $o(T)$-Nash equilibrium for the arms where a principal running $M$ receives at most $o(T)$ revenue. 
\end{theorem}

Here we assume the game is played under a \emph{tacit} observational model, meaning that arms can only observe which arms get pulled by the principal, but not how much value they give to the principal. In the \emph{explicit} observational model, where arms can see both which arms get pulled and how much value they pass on, even stronger results hold.

\begin{theorem}[informal restatement of Theorem \ref{thm:advexplicit}]\label{thm:introthm2}
Let $M$ be a low-regret algorithm for the classic multi-armed bandit problem with adversarially chosen values. Then there exists an instance of the strategic multi-armed bandit problem in the explicit observational model along with a $o(T)$-Nash equilibrium for the arms where a principal running $M$ receives zero revenue. 
\end{theorem}

While not immediately apparent from the above claims, these instances where low-regret algorithms fail are far from pathological; in particular, there is a problematic equilibrium for any instance where arm $i$ receives a fixed reward $v_i$ each round it is pulled, as long as the the gap between the largest and second-largest $v_i$ is not too large (roughly 1/\#arms). 

The driving cause behind both results is possible collusion between the arms (similar to collusion that occurs in the setting of repeated auctions, see \cite{SkrzypaczH04}). For example, consider a simple instance of this problem with two strategic arms, where arm 1 always gets private reward 1 if pulled and arm 2 always gets private reward 0.8. In this example, we also assume the principal is using algorithm EXP3. By always reporting some value slightly larger than 0.8, arm 1 can incentivize the principal to almost always pull it in the long run. This gains arm 1 roughly 0.2 utility per round (and arm 2 nothing). On the other hand, if arm 1 and arm 2 never pass along any surplus to the principal, they will likely be played equally often, gaining arm 1 roughly 0.5 utility per round and arm 2 0.4 utility per round. 

To show such a market-sharing strategy works for general low-regret algorithms, much more work needs to be done. The arms must be able to enforce an even split of the principal's pulls (as soon as the principal starts lopsidedly pulling one arm more often than the others, the remaining arms can defect and start reporting their full value whenever pulled). As long as the principal guarantees good performance in the non-strategic adversarial case (achieving $o(T)$ regret), we show that the arms can (at $o(T)$ cost to themselves) cooperate so that they are all played equally often.

\subsubsection{Mechanisms for strategic arms with stochastic values}

We next show that, in certain settings, it is in fact possible for the principal to extract positive values from the arms per round. We consider a setting where each arm $i$'s reward when pulled is drawn independently from some distribution $D_i$ with mean $\mu_i$ (known to arm $i$ but not to the principal). In this case the principal can extract the value of the second-best arm. In the below statement, we are using the term ``truthful mechanism'' quite loosely as shorthand for ``strategy that induces a game among the arms where each arm has a dominant strategy.''

\begin{theorem}[restatement of Corollary \ref{cor:stostra}]\label{thm:introthm3}
Let $\mu'$ be the second largest mean amongst the set of $\mu_i$s. Then there exists a truthful mechanism for the principal that guarantees revenue at least $\mu'T - o(T)$ when arms use their dominant strategies.
\end{theorem}

The mechanism in Theorem \ref{thm:introthm3} is a slight modification of the second-price auction strategy adapted to the multi-armed bandit setting. The principal begins by asking each arm $i$ for its mean $\mu_i$, where we incentivize arms to answer truthfully by recompensating arms according to a proper scoring rule. For the remainder of the rounds, the principal then asks the arm with the highest mean to give him the second-largest mean worth of value per round. If this arm fails to comply in any round, the principal avoids picking this arm for the remainder of the rounds. (A more detailed description of the mechanism can be seen in Mechanism \ref{mech:sto-stra-known-tacit} in Section \ref{sect:ub}). In addition, we show that the performance of this mechanism is as good as possible in this setting; no mechanism can do better than the second-best arm in the worst case (Lemma \ref{lem:secondbest}).

We further show how to adapt this mechanism in the setting where some arms are strategic and some arms are non-strategic (and our mechanism does not know which arms are which). 

\begin{theorem}[restatement of Theorem \ref{thm:stonstra}]
Let $\mu_s$ be the second largest mean amongst the means of the strategic arms, and let $\mu_{n}$ be the largest mean amongst the means of the non-strategic arms. Then there exists a truthful mechanism for the principal that guarantees (with probability $1-o(1/T)$) revenue at least $\max(\mu_{s}, \mu_{n})T - o(T)$ when arms use their dominant strategies.
\end{theorem}

A detailed description of the modified mechanism can be found in Mechanism \ref{mech:sto-nstra-known-tacit} in Section \ref{sect:ub}.

\subsection{Related work}\label{relatedwork}

The study of classical multi-armed bandit problems was initiated by \cite{Robbins52}, and has since grown into an active area of study. The most relevant results for our paper concern the existence of low-regret bandit algorithms in the adversarial setting, such as the EXP 3 algorithm (\cite{AuerCNS03}), which achieves regret $\tilde{O}(\sqrt{KT})$. Other important results in the classical setting include the upper confidence bound (UCB) algorithm for stochastic bandits (\cite{LaiR85}) and the work of \cite{GittinsJ74} for Markovian bandits. For further details about multi-armed bandit problems, see the survey \cite{Bubeck12}.

One question that arises in the strategic setting (and other adaptive settings for multi-armed bandits) is what the correct notion of regret is; standard notions of regret guarantee little, since the best overall arm may still have a small total reward. \cite{AroraDT12} considered the multi-armed bandit problem with an adaptive adversary  and introduced the quantity of ``policy regret'', which takes the adversary's adaptiveness into account. They showed that any multi-armed bandit algorithm will get $\Omega(T)$ policy regret. This indicates that it is not enough to treat strategic behaviors as an instance of adaptively adversarial behavior; good mechanisms for the strategic multi-armed bandits problem must explicitly take advantage of the rational self-interest of the arms.

Our model bears some similarities to the principal-agent problem of contract theory, where a principal employs an more informed agent to make decisions on behalf of the principal, but where the agent may have incentives misaligned from the principal's interests when it gets private savings (for example \cite{Chassang13}). For more details on principal-agent problem, see the book \cite{LaffontM2002}. Our model can be thought of as a sort of multi-armed version of the principal-agent problem, where the principal has many agents to select from (the arms) and can try to use competition between the agents to align their interests with the principal. 

Our negative results are closely related to results on collusions in repeated auctions. Existing theoretical work \cite{McAfeeM92,AtheyB01,JohnsonR99,Aoyagi03,Aoyagi07,SkrzypaczH04} has shown that collusive schemes exist in repeated auctions in many different settings, e.g., with/without side payments, with/without communication, with finite/infinite typespace. In some settings, efficient collusion can be achieved, i.e., bidders can collude to allocate the good to the bidders who values it the most and leave 0 asymptotically to the seller. Even without side payments and communication, \cite{SkrzypaczH04} showed that tacit collusion exists and can achieve asymptotic efficiency with a large cartel. 

Our truthful mechanism uses a proper scoring rule \cite{Brier50,McCarthy56} implicitly. In general, scoring rules are used to assessing the accuracy of a probabilistic prediction. In our mechanisms, we use a logarithmic scoring rule to incentivize arms to truthfully report their average rewards. 

Our setting is similar to settings considered in a variety of work on dynamic mechanism design, often inspired by online advertising. 
\cite{BergemannV96} considers the problem where a buyer wants to buy a stream of goods with an unknown value from two sellers, and examines Markov perfect equilibria in this model. \cite{Babaioff09,DevanurK09, Babaioff10} study truthful pay-per-click auctions where the auctioneer wishes to design a truthful mechanism that maximizes the social welfare.  \cite{Kremer14,FrazierKKK14} consider the scenario where the principal cannot directly choose which arm to pull, and instead must incentivize a stream of strategic players to prevent them from acting myopically. \cite{AminRS13, AminRS14} consider a setting where a seller repeatedly sells to a buyer with unknown value distribution, but the buyer is more heavily discounted than the seller. \cite{KakadeLN13} develops a general method for finding optimal mechanisms in settings with dynamic private information. \cite{NazerzadehSV08} develops an ex ante efficient mechanism for the Cost-Per-Action charging scheme in online advertising. 

\subsection{Open Problems and Future Directions}

We are far from understanding the complete picture of multi-armed bandit problems in strategic settings. Many questions remain, both in our model and related models. 

One limitation of our negative results is that they only show there exists some `bad' approximate Nash equilibrium for the arms, i.e., one where any low-regret principal receives little revenue. This, however, says nothing about the space of all approximate Nash equilibria. Does there exist a low-regret mechanism for the principal along with an approximate Nash equilibria for the arms where the principal extracts significant utility? An affirmative answer to this question would raise hope for the possibility of a mechanism that can perform well in both the adversarial and strategic setting, whereas a negative answer would strengthen our claim that these two settings are fundamentally at odds.

One limitation of our positive results is that all of the learning takes place at the beginning of the protocol, and is deferred to the arms themselves. As a result, our mechanism fails in cases where the arms' distributions can change over time. Is it possible to design good mechanisms for such settings? Ideally, any good mechanism should learn the arms' values continually throughout the $T$ rounds, but this seems to open up the possibility of collusion between the arms.

Throughout this paper, whenever we consider strategic bandits we assume their rewards are stochastically generated. Can we say anything about strategic bandits with adversarially generated rewards? The issue here seems to be defining what a strategic equilibrium is in this case - arms need some underlying priors to reason about their future expected utility. One possibility is to consider what happens when the arms all play no-regret strategies with respect to some broad class of strategies.

Finally, there are other quantities one may wish to optimize instead of the utility of the principal. For example, is it possible to design an efficient principal, who almost always picks the best arm (even if the arm passes along little to the principal)? Theorem \ref{thm:advtacit} implies the answer is no if the principal also has to be efficient in the adversarial case, but are there other models where we can answer this question affirmatively?

\section{Preliminaries}
\label{sec:prelim}

\subsection{Classic Multi-Armed Bandits}

We begin by reviewing the definition of the classic multi-armed bandits problem and associated quantities.

In the classic multi-armed bandit problem a learner (the \textit{principal}) chooses one of $K$ choices (arms) per round, over $T$ rounds. On round $t$, the principal receives some reward $v_{i,t} \in [0, 1]$ for pulling arm $i$. The values $v_{i,t}$ are either drawn independently from some distribution corresponding to arm $i$ (in the case of \textit{stochastic bandits}) or adaptively chosen by an adversary (in the case of \textit{adversarial bandits}). Unless otherwise specified, we will assume we are in the adversarial setting.

Let $I_t$ denote the arm pulled by the principal at round $t$. The \textit{revenue} of an algorithm $M$ is the random variable 
\[
\Rev(M) = \sum_{t=1}^T v_{I_t,t}
\]
\noindent
and the the \textit{regret} of $M$ is the random variable
\[
\Reg(M) = \max_i  \sum_{t=1}^T v_{i,t} - \Rev(M)
\]

\begin{definition}[$\delta$-Low Regret Algorithm]
Mechanism $M$ is a \textit{$\delta$-low regret algorithm} for the multi-armed bandit problem if 
\[
\E[\Reg(M)] \leq  \delta.
\]

Here the expectation is taken over the randomness of $M$ and the adversary.
\end{definition}

\begin{definition}[($\rho, \delta$)-Low Regret Algorithm]
Mechanism $M$ is a \textit{($\rho, \delta$)-low regret algorithm} for the multi-armed bandit problem if with probability $1 - \rho$, 
\[
\Reg(M) \leq  \delta.
\]
\end{definition}

There exist $O(\sqrt{KT\log K})$-low regret algorithms and $(\rho, O(\sqrt{KT\log(K/\rho)}))$-low regret algorithms for the multi-armed bandit problem; see Section 3.2 of \cite{Bubeck12} for details. 

\subsection{Strategic Multi-Armed Bandits}

The strategic multi-armed bandits problem builds upon the classic multi-armed bandits problem with the notable difference that now arms are strategic agents with the ability to withhold some payment from the principal. Instead of the principal directly receiving a reward $v_{i,t}$ when choosing arm $i$, now arm $i$ receives this reward and passes along some amount $w_{i,t}$ to the principal, gaining the remainder $v_{i,t} - w_{i,t}$ as utility. 

For simplicity, in the strategic setting, we will assume the rewards $v_{i,t}$ are generated stochastically; that is, each round, $v_{i,t}$ is drawn independently from a distribution $D_i$ (where the distributions $D_i$ are known to all arms but not to the principal). While it is possible to pose this problem in the adversarial setting (or other more general settings), this comes at the cost of there being no clear notion of strategic equilibrium for the arms.

This strategic variant comes with two additional modeling assumptions. The first is the informational model of this game; what information does an arm observe when some other arm is pulled. We define two possible observational models:

\begin{enumerate}
\item \textbf{Explicit:} After each round $t$, every arm sees the arm played $I_t$ along with the quantity $w_{I_t,t}$ reported to the principal. 
\item \textbf{Tacit:} After each round $t$, every arm only sees the arm played $I_t$. 
\end{enumerate}

In both cases, only arm $i$ knows the size of the original reward $v_{i,t}$; in particular, the principal also only sees the value $w_{i,t}$ and learns nothing about the amount withheld by the arm. Collusion between arms is generally easier in the explicit observational model than in the tacit observational model.

The second modeling assumption is whether to allow arms to go into debt while paying the principal. In the \textit{restricted payment} model, we impose that $w_{i,t} \leq v_{i,t}$; an arm cannot pass along more than it receives in a given round. In the \textit{unrestricted payment} model, we let $w_{i,t}$ be any value in $[0,1]$. We prove our negative results in the restricted payment model and our positive results in the unrestricted payment model, but our proofs for our negative results work in both models (in particular, it is easier to collude and prove negative results in the unrestricted payment model).

Finally, we proceed to define the set of strategic equilibria for the arms. We assume the mechanism $M$ of the principal is fixed ahead of time and known to the $K$ arms. If each arm $i$ is using a (possibly adaptive) strategy $S_i$, then the expected utility of arm $i$ is defined as

\[
u_i(M, S_1,\dots,S_K) = \E\left[ \sum_{t=1}^T (v_{i,t} - w_{i,t}) \cdot \1_{I_t = i}\right].
\]

An $\varepsilon$-Nash equilibrium for the arms is then defined as follows. 

\begin{definition}[$\varepsilon$-Nash Equilibrium for the arms] 
Strategies $(S_1,...,S_K)$ form an $\varepsilon$-Nash equilibrium for the strategic multi-armed bandit problem if for all $i \in [n]$ and any deviating strategy $S_i'$, 
\[
u_i(S_1, \dots ,S_i, \dots, S_K) \geq u_i(S_1,\dots,S'_i,\dots, S_K) - \varepsilon.
\]
\end{definition}

The goal of the principal is to choose a mechanism $M$ which guarantees large revenue in any $\varepsilon$-Nash Equilibrium for the arms.

In Section \ref{sect:ub}, we will construct mechanisms for the strategic multi-armed bandit problem which are truthful for the arms. We define the related terminology below. 

\begin{definition}[Dominant Strategy]
When the principal uses mechanism $M$, we say $S_i$ is a dominant strategy for arm $i$ if for any deviating strategy $S_i'$ and any strategies for other arms $S_1,..,S_{i-1},S_{i+1},...,S_K$, 
\[
u_i(M, S_1, \dots ,S_i, \dots, S_K) \geq u_i(M, S_1,\dots,S'_i,\dots, S_K).
\]
\end{definition}
\begin{definition}[Truthfulness]
We say that a mechanism $M$ for the principal is truthful, if all arms have some dominant strategies. 
\end{definition}

\section{Negative Results}\label{sect:negative}

In this section, we show that algorithms that achieve low-regret in the multi-armed bandits problem with adversarial values perform poorly in the strategic multi-armed bandits problem. Throughout this section, we will assume we are working in the restricted payment model (i.e., arms can only pass along a value $w_{i,t}$ that is at most $v_{i,t}$), but all proofs work also work in the unrestricted payment model (and in fact are much easier there). 

\subsection{Explicit Observational Model}

We begin by showing that in the explicit observational model, there is an approximate equilibrium for the arms that results in the principal receiving no revenue. Since arms can view other arms' reported values, it is easy to collude in the explicit model; simply defect and pass along the full amount as soon as you observe another arm passing along a positive amount. 

\begin{theorem}
\label{thm:advexplicit}
Let mechanism $M$ be a $\delta$-low regret algorithm for the multi-armed bandit problem. Then in the strategic multi-armed bandit problem under the explicit observational model, there exist distributions $D_i$ and a $(\delta+1)$-Nash equilibrium for the arms where a principal using mechanism $M$ receives zero revenue.
\end{theorem}

\begin{proof}
Consider the two-arm setting where $D_1$ and $D_2$ are both deterministic distributions supported entirely on $\{1\}$, so that $v_{i,t} = 1$ for all $i =1,2$ and $t \in [T]$. Consider the following strategy $S^*$ for arm $i$:

\begin{enumerate}
\item Set $w_{i,t} = 0$ if at time $1,...,t-1$, the other arm always reports 0 when pulled. 
\item Set $w_{i,t} = 1$ otherwise. 
\end{enumerate}

We will show that $(S^*,S^*)$ is a $(\delta+1)$-Nash Equilibrium. It suffices to show that arm 1 can get at most $\delta + 1$ more utility by deviating. Consider any deviating strategy $S'$ for arm 1. By convexity, we can assume $S'$ is deterministic (there is some best deterministic deviating strategy). Since mechanism $M$ might be randomized, let $R$ be the randomness used by $M$ and define $M_R$ to be the deterministic mechanism when $M$ uses randomness $R$. Now, consider the case when arm 1 plays strategy $S'$, arm 2 plays strategy $S^*$ and the principal is usings mechanism $M_R$. 

\begin{enumerate}
\item If arm 1 never reports any value larger than 0 when pulled, then $S'$ behaves exactly the same as $S^*$. Therefore, 
\[
u_1(M_R,S',S^*) = u_1(M_R,S^*,S^*).
\]
\item If arm 1 ever reports some value larger than 0 when pulled, let $\tau_R$ be the first time it does so. We know that $S'$ behaves the same as $S^*$ before $\tau_R$. Therefore,
\begin{eqnarray*}
u_1(M_R,S',S^*) &\leq& u_1(M_R,S^*,S^*) + \sum_{t=\tau_R}^T (v_{1,t} - w_{1,t}) \cdot \1_{I_t = 1}\\
&\leq& u_1(M_R,S^*,S^*) + 1+  \sum_{t=\tau_R+1}^T (\max(w_{1,t}, w_{2,t}) - w_{1,t}) \cdot \1_{I_t = 1}\\
\end{eqnarray*}
\end{enumerate}
So in general, we have
\[
u_1(M_R,S',S^*) \leq u_i(M_R,S^*,S^*) + 1+  \sum_{t=\tau_R+1}^T (\max(w_{1,t}, w_{2,t}) - w_{1,t}) \cdot \1_{I_t = 1}.
\]
Therefore
\begin{eqnarray*}
u_1(M,S',S^*) &=& \E_R[u_1(M_R,S',S^*)] \\
 &\leq& \E_R[u_1(M_R,S^*,S^*)] + 1 + \E_R\left[\sum_{t=\tau_R+1}^T (\max(w_{1,t}, w_{2,t}) - w_{1,t}) \cdot \1_{I_t = 1}\right] \\
 &=& u_1(M,S^*,S^*)  + 1 + \E_R\left[\sum_{t=\tau_R+1}^T (\max(w_{1,t}, w_{2,t}) - w_{1,t}) \cdot \1_{I_t = 1}\right]. \\
\end{eqnarray*}

Notice that this expectation is at most the regret of $M$ in the classic multi-armed bandit setting when the adversary sets rewards equal to the values $w_{1,t}$ and $w_{2,t}$ passed on by the arms when they play $(S', S^*)$. Therefore, by our low-regret guarantee on $M$, we have that
\[
\E_R \left[\sum_{t=\tau_R+1}^T (\max(w_{1,t}, w_{2,t}) - w_{1,t}) \cdot \1_{I_t = 1}\right] \leq \delta. 
\]
Thus
\[
u_1(M,S',S^*) \leq u_1(M,S^*,S^*) + 1 +\delta
\]
and this is a $(1+\delta)$-approximate Nash equilibrium. Finally, it is easy to check that the principal receives zero revenue when both arms play according to this equilibrium strategy. 
\end{proof}

\subsection{Tacit Observational Model}

We next show that even in the tacit observational model, where the arms don't see the amounts passed on by other arms, it is still possible for the arms to collude and leave the principal with $o(T)$ revenue. The underlying idea here is that the arms work to try to maintain an equal market share, where each of the $K$ arms are each played approximately $1/K$ of the time. To ensure this happens, arms collude so that arms that aren't as likely to be pulled pass along a tiny amount $\epsilon$ to the principal, whereas arms that have been pulled a lot or are more likely to be pulled pass along $0$; this ends up forcing any low-regret algorithm for the principal to choose all the arms equally often. Interestingly, unlike the collusion strategy in the explicit observational model, this collusion strategy is \emph{mechanism dependent}, as arms need to estimate the probability they will be pulled in the next round. 

We begin by proving this result for the case of two arms, where the proof is slightly simpler. 

\begin{theorem}
\label{thm:advtacit2arms}
Let mechanism $M$ be a $\left(\rho, \delta\right)$-low regret algorithm for the multi-armed bandit problem with two arms, where $\rho \leq T^{-2}$ and $\delta \geq \sqrt{T\log T}$. Then in the strategic multi-armed bandit problem under the tacit observational model, there exist distributions $D_1, D_2$ and an $O(\sqrt{T\delta})$-Nash Equilibrium where a principal using mechanism $M$ gets at most $O(\sqrt{T\delta})$ revenue. 
\end{theorem}

\begin{proof}
Let $D_1$ and $D_2$ be distributions with means $\mu_1$ and $\mu_2$ respectively, such that $|\mu_1-\mu_2| \leq \max(\mu_1,\mu_2)/2$. Additionally, assume both $D_1$ and $D_2$ are supported on $[\sqrt{\delta/T}, 1]$.  We now describe the equilibrium strategy $S^*$ (the below description is for arm 1; $S^*$ for arm 2 is symmetric):
\begin{enumerate}
\item Set parameters $B = 6\sqrt{T\delta}$ and $\theta = \sqrt{\frac{\delta}{T}}$. 
\item Define $c_{1,t}$ to be the number times arm $1$ is pulled in rounds $1,...,t$. Similarly define $c_{2,t}$ to be the number times arm $2$ is pulled in rounds $1,...,t$.
\item For $t = 1,\dots,T$:
\begin{enumerate}
\item If there exists a $t' \leq t-1$ such that $c_{1,t'} < c_{2,t'} - B$, set $w_{1,t} = v_{1,t}$.
\item If the condition in (a) is not true, let $p_{1,t}$ be the probability that the principal will pick arm 1 in this round conditioned on the history (assuming player $2$ is also playing $S^*$), and let $p_{2,t} = 1-p_{1,t}$. Then:
\begin{enumerate}
\item If $c_{1,t-1} < c_{2,t-1}$ and $p_{1,t}<p_{2,t}$, set $w_{1,t} = \theta$.
\item Otherwise, set $w_{1,t} = 0$.
\end{enumerate}
\end{enumerate}
\end{enumerate}

We will now show that $(S^*,S^*)$ is an $O(\sqrt{T\delta})$-Nash equilibrium. To do this, for any deviating strategy $S'$, we will both lower bound $u_1(M, S^*, S^*)$ and upper bound $u_1(M, S', S^*)$, hence bounding the net utility of deviation.

We begin by proving that $u_1(M, S^*, S^*) \geq \frac{\mu_2 T}{2} - O(\sqrt{T\delta})$. We need the following lemma.
\begin{lemma}
\label{lem:2eq}
If both arms are using strategy $S^*$, then with probability $\left(1-\frac{4}{T}\right)$, $|c_{1,t} -c_{2,t}| \leq B$ for all $t\in[T]$. 
\end{lemma}

\begin{proof}
Assume that both arms are playing the strategy $S^*$ with the modification that they never defect (i.e. condition (a) in the above strategy is removed). This does not change the probability that $|c_{1,t} - c_{2,t}| \leq B$ for all $t \in [T]$. 

Define $R_{1,t} = \sum_{s=1}^{t}w_{1,s} - \sum_{s=1}^{t}w_{I_s, s}$ be the regret the principal experiences for not playing only arm 1. Define $R_{2,t}$ similarly. We will begin by showing that with high probability, these regrets are bounded both above and below. In particular, we will show that with probability at least $1-\frac{2}{T}$, $R_{i,t}$ lies in $[-2\theta\sqrt{T\log T} - \delta, \delta]$ for all $t \in [T]$ and $i \in \{1, 2\}$.

To do this, note that there are two cases where the regrets $R_{1,t}$ and $R_{2,t}$ can possibly change. The first is when $p_{1,t}>p_{2,t}$ and $c_{1,t}>c_{2,t}$. In this case, the arms offer $(w_{1,t}, w_{2,t}) = (0, \theta)$. With probability $p_{1,t}$ the principal chooses arm $1$ and the regrets update to $(R_{1,t+1}, R_{2,t+1}) = (R_{1,t}, R_{2,t} + \theta)$, and with probability $p_{2,t}$ the principal chooses arm $2$ and the regrets update to $(R_{1,t+1}, R_{2, t+1}) = (R_{1, t}-\theta, R_{2,t})$. It follows that $\E[R_{1,t+1}+R_{2,t+1}|R_{1,t}+R_{2,t}] = R_{1,t}+R_{2,t} + (p_{1,t}-p_{2,t})\theta \geq R_{1,t} + R_{2,t}$. 

In the second case, $p_{1,t}<p_{2,t}$ and $c_{2,t}<c_{1,t}$, and a similar calculation shows again that $\E[R_{1,t+1}+R_{2,t+1}|R_{1,t}+R_{2,t}] = R_{1,t}+R_{2,t} + (p_{2,t}-p_{1,t})\theta \geq R_{1,t} + R_{2,t}$. It follows that $R_{1,t} + R_{2,t}$ forms a submartingale.

From the above analysis, it is also clear that $\left|(R_{1,t+1} + R_{2,t+1}) - (R_{1,t} + R_{2,t})\right| \leq \theta$. It follows from Azuma's inequality that, for any fixed $t \in [T]$,

$$\Prob\left[R_{1,t} + R_{2,t} \leq -2\theta\sqrt{T\log T}\right] \leq \frac{1}{T^2}$$

Applying the union bound, with probability at least $1-\frac{1}{T}$, $R_{1,t} + R_{2,t} \geq -2\theta\sqrt{T\log T}$ for all $t \in [T]$. Furthermore, since the principal is using a $\left(T^{-2}, \delta\right)$-low-regret algorithm, it is also true that with probability at least $1-T^{-2}$ (for any fixed $t$) both $R_{1,t}$ and $R_{2,t}$ are at most $\delta$. Applying the union bound again, it is true that $R_{1,t} \leq \delta$ and $R_{2,t} \leq \delta$ for all $t$ with probability at least $1-\frac{1}{T}$. Finally, combining this with the earlier inequality (and applying union bound once more), with probability at least $1 - \frac{2}{T}$, $R_{i,t} \in \left[-2\theta\sqrt{T\log T}-\delta, \delta\right]$, as desired. For the remainder of the proof, condition on this being true.

We next proceed to bound the probability that (for a fixed $t$) $c_{1,t} - c_{2,t} \leq B$. Define the random variable $\tau$ to be the largest value $s \leq t$ such that $c_{1,\tau} - c_{2, \tau} = 0$ -- note that if $c_{1,t} - c_{2,t} \geq 0$, then $c_{1,s} - c_{2,s} \geq 0$ for all $s$ in the range $[\tau, t]$. Additionally let $\Delta_{s}$ denote the $\pm 1$ random variable given by the difference $(c_{1,s} - c_{2,s})-(c_{1,s-1} - c_{2,s-1})$. We can then write

\begin{eqnarray*}
c_{1,t} - c_{2,t} &\leq & \sum_{s = \tau+1}^{t} \Delta_{s} \\
&\leq & \sum_{s=\tau+1}^{t} \Delta_{s}\cdot \1_{p_{1,s}>p_{2,s}} + \sum_{s=\tau+1}^{t} \Delta_{s}\cdot \1_{p_{1,s}\leq p_{2,s}}
\end{eqnarray*}

Here the first summand corresponds to times $s$ where one of the arms offers $\theta$ (and hence the regrets change), and the second summand corresponds to times where both arms offer $0$. Note that since $c_{1,s} \geq c_{2,s}$ in this interval, the regret $R_{2,s}$ increases by $\theta$ whenever $\Delta_{s} = 1$ (i.e., arm $1$ is chosen), and furthermore no choice of arm can decrease $R_{2,s}$ in this interval. Since we know that $R_{2,s}$ lies in the interval $\left[-2\theta\sqrt{T\log T}-\delta, \delta\right]$ for all $s$, this bounds the first sum by

$$\sum_{s=\tau+1}^{t} \Delta_{s}\cdot \1_{p_{1,s}>p_{2,s}} \leq \frac{2\delta + 2\theta\sqrt{T\log T}}{\theta} = \frac{2\delta}{\theta} + 2\sqrt{T\log T}$$

On the other hand, when $p_{1,s} \leq p_{2,s}$, then $\E[\Delta_{s}] = p_{1,s}-p_{2,s} \leq 0$. By Hoeffding's inequality, it then follows that with probability at least $1 - \frac{1}{T^2}$, 

$$\sum_{s=\tau+1}^{t} \Delta_{s}\cdot \1_{p_{1,s}\leq p_{2,s}} \leq 2\sqrt{T\log T}$$

\noindent
Altogether, this shows that with probability at least $1 - \frac{1}{T^2}$,

$$c_{1,t} - c_{2,t} \leq \frac{2\delta}{\theta} + 4\sqrt{T\log T} \leq 6\sqrt{T\delta} = B$$

The above inequality therefore holds for all $t$ with probability at least $1 - \frac{1}{T}$. Likewise, we can show that $c_{2,t} - c_{1,t} \leq B$ also holds for all $t$ with probability at least $1 - \frac{1}{T}$. Since we are conditioned on the regrets $R_{i,t}$ being bounded (which is true with probability at least $\frac{2}{T}$), it follows that $|c_{1,t} - c_{2,t}| \leq B$ for all $t$ with probability at least $1 - \frac{4}{T}$. 

\end{proof}

By Lemma \ref{lem:2eq}, we know that with probability $1-\frac{4}{T}$, $|c_{1,t} -c_{2,t}| \leq B$ throughout the mechanism. In this case, arm 1 never uses step (a), and $c_{1,T} \geq (T-B)/2$. Therefore 

\begin{eqnarray*}
u_1(M, S^*, S^*) &\geq&  \left(1-\frac{4}{T}\right) \cdot (\mu_1-\theta) \cdot (T-B)/2 \\
&\geq & \frac{\mu_1T}{2}\left(1 - \frac{4}{T} - \frac{\theta}{\mu_1} - \frac{B}{T}\right) \\
&=& \frac{\mu_1T}{2} - 2\mu_1 - \frac{\theta T}{2} - \frac{B\mu_1}{2} \\
&\geq & \frac{\mu_1T}{2} - O(\sqrt{T\delta})
\end{eqnarray*}

Now we will show that $u_1(M,S',S^*) \leq \frac{\mu_1T}{2} + O(\sqrt{T\delta})$. Without loss of generality, we can assume $S'$ is deterministic. Let $M_R$ be the deterministic mechanism when $M$'s randomness is fixed to some outcome $R$. Consider the situation when arm $1$ is using strategy $S'$, arm 2 is using strategy $S^*$ and the principal is using mechanism $M_R$. There are two cases:
\begin{enumerate}
\item $c_{1,t} -c_{2,t} \leq B$ is true for all $t\in[T]$. In this case, we have 
\[
u_1(M_R,S',S^*) \leq c_{1,T} \cdot \mu_1 \leq \mu_1(T+B)/2.
\]
\item There exists some $t$ such that $c_{1,t} -c_{2,t} > B$: Let $\tau_R+1$ be the smallest $t$ such that $c_{1,t} -c_{2,t} > B$. We know that $c_{1,\tau_R} -c_{2,\tau_R} \leq B$. Therefore we have
\begin{eqnarray*}
u_1(M_R,S',S^*) &=&\sum_{t=1}^T (\mu_1 - w_{1,t}) \cdot \1_{I_t = 1} \\
&=& \sum_{t=1}^T (\mu_1 - w_{2,t}) \cdot \1_{I_t = 1} + \sum_{t=1}^T (w_{2,t} - w_{1,t}) \cdot \1_{I_t = 1} \\
&\leq& c_{1,\tau_R} \mu_1 + \mu_1 + (T-\tau_R-1) \max(\mu_1-\mu_2,0) + \sum_{t=1}^T (w_{2,t} - w_{1,t}) \cdot \1_{I_t = 1} \\
&\leq& \mu_1(\tau_R+B)/2 + \mu_1 + (T-\tau_R-1) (\mu_1/2) +\sum_{t=1}^T (w_{2,t} - w_{1,t}) \cdot \1_{I_t = 1} \\
&\leq& \mu_1T/2 + \mu_1(B+1)/2 + \sum_{t=1}^T (w_{2,t} - w_{1,t}) \cdot \1_{I_t = 1}. \\
\end{eqnarray*}
\end{enumerate}
In general, we thus have that
\[
u_1(M_R,S',S^*) \leq \mu_1T/2 + \mu_1(B+1)/2 + \max\left(0, \sum_{t=1}^T (w_{2,t} - w_{1,t}) \cdot \1_{I_t = 1}\right). \\
\]
Therefore
\begin{eqnarray*}
u_1(M,S',S^*) &=& \E_R[u_1(M_R,S',S^*)] \\
&\leq& \mu_1T/2 + \mu_1(B+1)/2 + \E_R\left[ \max\left(0, \sum_{t=1}^T (w_{2,t} - w_{1,t}) \cdot \1_{I_t = 1}\right)\right]. \\ 
\end{eqnarray*}
Notice that $\sum_{t=1}^T (w_{2,t} - w_{1,t}) \cdot \1_{I_t = 1}$ is the regret of not playing arm 2 (i.e., $R_2$ in the proof of Lemma \ref{lem:2eq}). Since the mechanism $M$ is $(\rho, \delta)$ low regret, with probability $1-\rho$, this sum is at most $\delta$ (and in the worst case, it is bounded above by $T \mu_2$). We therefore have that:

\begin{eqnarray*}
u_1(M,S',S^*) &\leq& \frac{\mu_1T}{2} + \frac{\mu_1(B+1)}{2} + \delta + \rho T \mu_2 \\
&\leq & \frac{\mu_1T}{2} + O(\sqrt{T\delta})
\end{eqnarray*}

From this and our earlier lower bound on $u_1(M, S^*, S^*)$, it follows that $u_1(M, S',S^*) - u_1(M, S^*, S^*) \leq O(\sqrt{T\delta})$, thus establishing that $(S^*, S^*)$ is an $O(\sqrt{T\delta})$-Nash equilibrium for the arms.

Finally, to bound the revenue of the principal, note that if the arms both play according to $S^*$ and $|c_{1,t} - c_{2,t}| \leq B$ for all $t$ (so they do not defect), the principal gets a maximum of $T\theta = O(\sqrt{T\delta})$ revenue overall. Since (by Lemma \ref{lem:2eq}) this happens with probability at least $1 - \frac{4}{T}$ (and the total amount of revenue the principal is bounded above by $T$), it follows that the total expected revenue of the principal is at most $O(\sqrt{T\delta})$. 
\end{proof}


We now extend this proof to the $K$ arm case, where $K$ can be as large as $T^{1/3}/\log(T)$. 

\begin{theorem}\label{thm:advtacit}
Let mechanism $M$ be a $\left(\rho, \delta\right)$-low regret algorithm for the multi-armed bandit problem with $K$ arms, where $K \leq T^{1/3}/\log(T)$, $\rho \leq T^{-2}$, and $\delta \geq \sqrt{T\log T}$. Then in the strategic multi-armed bandit problem under the tacit observational model, there exist distributions $D_i$ and an $O(\sqrt{KT\delta})$-Nash Equilibrium for the arms where the principal gets at most $O(\sqrt{KT\delta})$ revenue. 
\end{theorem}

\begin{proof}[Proof Sketch]
As in the previous proof, let $\mu_i$ denote the mean of the $i$th arm's distribution $D_i$. Without loss of generality, further assume that $\mu_{1} \geq \mu_{2} \geq \dots \geq \mu_{K}$. We will show that as long as $\mu_{1}-\mu_{2} \leq \frac{\mu_1}{K}$, there exists some $O(\sqrt{KT\delta})$-Nash equilibrium for the arms where the principal gets at most $O(\sqrt{KT\delta})$ revenue.

We begin by describing the equilibrium strategy $S^*$ for the arms. Let $c_{i,t}$ denote the number of times arm $i$ has been pulled up to time $t$. As before, set $B = 7\sqrt{KT\delta}$ and set $\theta = \sqrt{\frac{K\delta}{T}}$. The equilibrium strategy for arm $i$ at time $t$ is as follows:

\begin{enumerate}
\item
If at any time $s\leq t$ in the past, there exists an arm $j$ with $c_{j, s} - c_{i, s} \geq B$, defect and offer your full value $w_{i,t} = \mu_{i}$.
\item
Compute the probability $p_{i,t}$, the probability that the principal will pull arm $i$ conditioned on the history so far. 
\item
Offer $w_{i,t} = \theta(1-p_{i,t})$. 
\end{enumerate} 

The remainder of the proof proceeds similarly as the proof of Theorem \ref{thm:advtacit2arms}. The full proof can be found in Appendix \ref{sect:applb}.
\end{proof}

While the theorems above merely claim that a bad set of distributions for the arms exists, note that the proofs above show it is possible to collude in a wide range of instances - in particular, any set of distributions which satisfy $\mu_1 - \mu_2 \leq \mu_1/K$. A natural question is whether we can extend the above results to show that it is possible to collude in any set of distributions. 

One issue with the collusion strategies in the above proofs is that if $\mu_1 - \mu_2 > \mu_1/K$, then arm 1 will have an incentive to defect in any collusive strategy that plays all the arms evenly (arm 1 can report a bit over $\mu_2$ per round, and make $\mu_1 - \mu_2$ every round instead of $\mu_1$ every $K$ rounds). One solution to this is to design a collusive strategy that plays some arms more than others in equilibrium (for example, playing arm $1$ 90\% of the time). We show how to modify our result for two arms to achieve an arbitrary market partition and thus work over a broad set of distributions. 

\begin{theorem}
\label{thm:alldists2arms}
Let mechanism $M$ be a $\left(\rho, \delta\right)$-low regret algorithm for the multi-armed bandit problem with two arms, where $\rho \leq T^{-2}$ and $\delta \geq \sqrt{T\log T}$. Then, in the strategic multi-armed bandit problem under the tacit observational model, for \emph{any} distributions $D_1, D_2$ of values for the arms (supported on $[\sqrt{\delta/T}, 1]$), there exists an $O(\sqrt{T\delta})$-Nash Equilibrium for the arms where a principal using mechanism $M$ gets at most $O(\sqrt{T\delta})$ revenue. 
\end{theorem}
\begin{proof}
See Appendix \ref{sect:applb}.
\end{proof}

Unfortunately, it as not as easy to modify the proof of Theorem \ref{thm:advtacit} to prove the same result for $K$ arms. It is an interesting open question whether there exist collusive strategies for $K$ arms that can achieve an arbitrary partition of the market. 

\section{Positive Results}\label{sect:ub}

In this section we will show that, in contrast to the previous results on collusion, there exists a mechanism for the principal that can obtain $\Theta(T)$ revenue from the arms. This mechanism essentially incentivizes each arm to report the mean of its distribution and then runs a second-price auction, asking the arm with the highest mean for the second-highest mean each round. By slightly modifying this mechanism, we can obtain a mechanism that works for a combination of strategic and non-strategic arms.

Throughout this section we will assume we are working in the tacit observational model and the unrestricted payment model.

\subsection{All Strategic Arms with Stochastic Values}
We begin by considering the case when all arms are strategic. 

Define $\mu_i$ as the mean of distribution $D_i$ for $i=1,\dots,K$ and $u = -\log \left( \min_{i:\mu_i \neq 0} \mu_i \right) +1$. We assume throughout that $u = o(T/K)$. 
\begin{algorithm}[ht]

Play each arm once (i.e. play arm 1 in the first round, arm 2 in the second round, etc.). Let $w_{i}$ be the value arm $i$ reports in round $i$.

Let $i^* = \arg\max w_{i}$ (breaking ties lexicographically), and let $w' = \max_{i\neq i^*} w_{i}$.

Tell arm $i^*$ the value of $w'$. Play arm $i^*$ for $R = T - (u+2)K - 1$ rounds. If arm $i^*$ ever reports a value different from $w'$, stop playing it immediately. If arm $i^*$ always gives $w'$, play it for one bonus round (ignoring the value it reports). 
	
For each arm $i$ such that $i\neq i^*$, play it for one round. 

For each arm $i$ satisfying $u + \log (w_{i}) \geq 0$, play it $\lfloor u + \log (w_{i}) \rfloor$ times. Then, with probability $u+ \log (w_{i})- \lfloor u+ \log (w_{i}) \rfloor$, play arm $i$ for one more round.

\SetAlgorithmName{Mechanism}
	\floatname{algorithm}{}
        \caption{Truthful mechanism for strategic arms with known stochastic values in the tacit model}
       \label{mech:sto-stra-known-tacit}
\end{algorithm}

We will first show that the dominant strategy of each arm in this mechanism includes truthfully reporting their mean at the beginning, and then then compute the principal's revenue under this dominant strategy. 

\begin{lemma} 
\label{lem:stra-dom}
The following strategy is the dominant strategy for arm $i$ in Mechanism \ref{mech:sto-stra-known-tacit}:
\begin{enumerate}
\item (line 1 of Mechanism \ref{mech:sto-stra-known-tacit}) Report the mean value $\mu_{i}$ of $D_i$ the first time when arm $i$ is played.
\item (lines 3,4 of Mechanism \ref{mech:sto-stra-known-tacit}) If $i = i^*$, for the $R$ rounds that the principal expects to see reported value $w'$, report the value $w'$. For the bonus round, report 0. If $i \neq i^*$, report 0. 
\item (line 5 of Mechanism \ref{mech:sto-stra-known-tacit}) For all other rounds, report $0$. 
\end{enumerate}
\end{lemma}

\begin{proof}
Note that the mechanism is naturally divided into three parts (in the same way the strategy above is divided into three parts): (1) the start, where each arm is played once and reports its mean, (2) the middle, where the principal plays the best arm and extracts the second-best arm's value (and plays each other arm once), and (3) the end, where the principal plays each arm some number of times, effectively paying them off for responding truthfully in step (1). To show the above strategy is dominant, we will proceed by backwards induction, showing that each part of the strategy is the best conditioned on an arbitrary history. 

We start with step (3). It is easy to check that these rounds don't affect how many times the arm is played or not. It follows that it is strictly dominant to just report 0 (and receive your full value for the turn). Note that the reward the arm receives in expectation for this round is $(u + \log(w_{i}))\mu_i$; we will use this later. 

For step (2), assume that $i=i^*$; otherwise, arm $i$ is played only once, and the dominant strategy is to report $0$ and receive expected reward $\mu_{i}$. Depending on what happened in step (1), there are two cases; either $w' \leq \mu_{i}$, or $w' > \mu_{i}$. We will show that if $w' \leq \mu_{i}$, the arm should play $w'$ for the next $R$ rounds (not defecting) and report 0 for the bonus round. If $w' > \mu_{i}$, the arm should play $0$ (defecting immediately).

Note that we can recast step (2) as follows: arm $i$ starts by receiving a reward from his distribution $D_{i}$. For the next $R$ turns, he can pay $w'$ for the privilege of drawing a new reward from his distribution (ending the game immediately if he refuses to pay). If $w' \leq \mu_{i}$, then paying for a reward $w'$ is positive in expectation, whereas if $w' > \mu_{i}$, then paying for a reward is negative in expectation. It follows that the dominant strategy is to continue to report $w'$ if $w' \leq \mu_{i}$ (receiving a total expected reward of $R(\mu_{i}-w') + \mu_{i}$) and to immediately defect and report $0$ if $w' > \mu_{i}$ (receiving a total expected reward of $\mu_{i}$).

Finally, we analyze step (1). We will show that, regardless of the values reported by the other players, it is a dominant strategy for arm $i$ to report its true mean $\mu_{i}$. If arm $i$ reports $w_{i}$, and $i \neq i^*$, then arm $i$ will receive in expectation reward

\begin{equation*}
G = (\mu_{i} - w_{i}) + \mu_{i} + \max(u + \log(w_{i}), 0)\mu_{i}
\end{equation*}

\noindent
If $u + \log(w_i) > 0$, then this is maximized when $w_{i} = \mu_{i}$ and $G= (u + \log(\mu_i) + 1)\mu_i$ (note that by our construction of $u$, $u + \log(\mu_i) \geq 1$). On the other hand, if $u + \log(w_i) \leq 0$, then this is maximized when $w_{i} = 0$ and $G = 2\mu_i$. Since $u+\log(\mu_i)+1 \geq 2$, the overall maximum occurs at $w_i = \mu_i$. 

Similarly, when arm $i$ reports $w_{i}$ and $i = i^*$, then arm $i$ receives in expectation reward

\begin{equation*}
G' = (\mu_{i}-w_{i}) + \min(0, R(\mu_{i}-w')) + \mu_{i} + \max(u + \log(w_{i}), 0)\mu_{i} 
\end{equation*}

\noindent
which is similarly maximized at $w_{i} = \mu_{i}$. Finally, it follows that if $\mu_{i} \leq w'$, $G=G'$, so it is dominant to report $w_{i} = \mu_{i}$. On the other hand, if $\mu_{i} > w'$, then reporting $w_{i} = \mu_{i}$ will ensure $i = i^*$ and so once again it is dominant to report $w_{i} = \mu_{i}$.
\end{proof}

\begin{corollary}
\label{cor:stostra}
Under Mechanism \ref{mech:sto-stra-known-tacit}, the principal will receive revenue at least $\mu'T - o(T)$ when arms use their dominant strategies, where $\mu'$ is the second largest mean in the set of means $\mu_{i}$. 
\end{corollary}

\begin{lemma}\label{lem:secondbest}
For any constant $\alpha > 0$, no truthful mechanism can guarantee $(\alpha \mu + (1-\alpha) \mu')T$ revenue in the worst case. Here $\mu$ is the largest value among $\mu_1,...,\mu_K$. And $\mu'$ is the second largest value among $\mu_1,...,\mu_K$.
\end{lemma}

\begin{proof}
Suppose there exists an truthful mechanism $A$ guarantees $(\alpha \mu + (1-\alpha) \mu')T$ revenue for any distributions. We will show this results in a contradiction. 

We now consider $L>\exp(1/\alpha)$ inputs. The $i$-th input has $\mu = b_i = 1/2 + i/(2L)$ and $\mu' = 1/2$. Among these inputs, one arm (call it arm $j^*$) is always the arm with largest mean and another arm is always the arm with the second largest mean. Other arms have the same input distribution in all the inputs. 

Consider all the arms are using their dominant strategies. For the $i$-th input, let $x_i T$ be the expected number of pulls by $A$ on the arm $k^*$ and $p_i T$ be the expected amount arm $k^*$ gives to the principal. Because the mechanism is truthful, in the $i$-th distribution, arm $k^*$ prefers its dominant strategy than the dominant strategy it uses in some $j$-th distribution ($i \neq j$). In other words, we have for $i \neq j$, 
\[
b_i x_i -p_i \geq b_i x_j -p_j. 
\]
We also have, for all $i$,
\[
b_i x_i -p_i \geq 0.
\]
By using these inequalities , we get for all $i$,
\[
p_i \leq b_i x_i + \sum_{j=1}^{i-1} x_j (b_{j+1} - b_j).
\]
On the other hand, $A$'s revenue in the $i$-th distribution is at most $(p_i + (1-x_i) \mu') T$. Therefore we have, for all $i$,
\[
p_i + (1- x_i) \mu' \geq \alpha \cdot b_i + (1-\alpha) \mu'.
\]
So we get
\[
(1- x_i) \mu' + b_i x_i + \sum_{j=1}^{i-1} x_j (b_{j+1} - b_j) \geq \alpha \cdot b_i + (1-\alpha) \mu'.
\]
It can be simplified as
\[
x_i \geq \alpha + \sum_{j=1}^{i-1} x_j \frac{b_{j+1} -b_j}{b_i-\mu'} =\alpha + \frac{1}{i} \cdot \sum_{j=1}^{i-1} x_j.
\]
By induction we get for all $i$, 
\[
x_i \geq \alpha \sum_{j=1}^i \frac{1}{i} > \alpha \ln(i).
\]
Therefore we have
\[
x_L > \alpha \ln(L) \geq 1.
\]
Here we get a contradiction.
\end{proof}

\begin{remark}
The above algorithm relies on the assumption that arms know their own means $\mu_i$. However, if the arms don't initially know their means, we can instead insert a phase at the beginning that lasts $T^{2/3}$ rounds where we pull each arm $T^{2/3}/K$ times and expect no reward to be passed on. This allows the arms to estimate their rewards, and the following phases can be appropriately adjusted to maintain a solution in $o(T)$-dominant strategies, losing an additional $O(T^{2/3})$ in revenue for the principal, but maintaining the revenue guarantee of $\mu_2 T - o(T)$. It is an interesting question whether a more clever stochastic bandit algorithm can be embedded without destroying dominant strategies, and also whether a solution exists in exact dominant strategies for this model. 
\end{remark}

\subsection{Strategic and Non-strategic Arms with Stochastic Values}
We now consider the case when some arms are strategic and other arms are non-strategic. Importantly, the principal does not know which arms are strategic and which are non-strategic. 

We define $\mu_i$ as the mean of distribution $D_i$ for $i=1,...,K$. Set $B = T^{2/3}$, $M=8T^{-1/3}\ln(KT)$ and $u = -\log \left( \min_{i:\mu_i \neq 0} \mu_i \right) +1+M$. We assume $u = o(\frac{T}{BK})$. 

\begin{algorithm}[ht]
Play each arm $B$ times (i.e. play arm 1 in the first $B$ rounds, arm 2 in the next $B$ rounds, etc.). Let $\bar{w}_{i}$ be the average value arm $i$ reported in its $B$ rounds.

Let $i^* = \arg\max \bar{w}_{i}$ (breaking ties lexicographically), and let $w' = \max_{i\neq i^*} \bar{w}_{i}$.

Tell arm $i^*$ the value of $w'$. Play arm $i^*$ for $R = T - (u+3)BK $ rounds. If arm $i^*$ ever report values with average less than $w'-M$ in any round after $B$ rounds in this step, stop playing it immediately. If arm $i^*$ gives average no less than $w'-M$, play it for $B$ bonus rounds (ignoring the value it reports). 

For each arm $i$ such that $i\neq i^*$, play it for $B$ rounds. 

For each arm $i$ satisfying $u + \log (\bar{w}_{i} - M) \geq 0$, play it $B\lfloor (u + \log (\bar{w}_{i}-M)) \rfloor$ times. Then, with probability $u+ \log (\bar{w}_{i}- M)- \lfloor u+ \log (\bar{w}_{i}-M) \rfloor$, play arm $i$ for $B$ more rounds. 
\SetAlgorithmName{Mechanism}
	\floatname{algorithm}{}
        \caption{Truthful mechanism for strategic/non-strategic arms in the tacit model}
       \label{mech:sto-nstra-known-tacit}
\end{algorithm}

\begin{lemma} 
\label{lem:nstra-dom}
The following strategy is the dominant strategy for arm $i$ in Mechanism \ref{mech:sto-nstra-known-tacit}:
\begin{enumerate}
\item (line 1 of Mechanism \ref{mech:sto-nstra-known-tacit}) For the first $B$ rounds, report a total sum of $(\mu_i +M) B$.
\item (lines 3,4 of Mechanism \ref{mech:sto-nstra-known-tacit}) If $i = i^*$, for the $R$ rounds that the principal expects to see reported value $w'$, report the value $w'-M$. For the $B$ bonus rounds, report 0. If $i \neq i^*$, report 0. 
\item (line 5 of Mechanism \ref{mech:sto-nstra-known-tacit}) For all other rounds, report $0$. 
\end{enumerate}
\end{lemma}

\begin{proof}
Similarly as the proof of Lemma \ref{lem:stra-dom}, the mechanism is divided into three parts: (1) the start, where each arm is played $B$ times and reports its mean, (2) the middle, where the principal plays the best arm and extracts the second-best arm's value (and plays each other arm $B$ times), and (3) the end, where the principal plays each arm some number of times, effectively paying them off for responding truthfully in step (1). To show the above strategy is dominant, we will proceed by backwards induction, showing that each part of the strategy is the best conditioned on an arbitrary history. 

For step (3), similarly as the proof of Lemma \ref{lem:stra-dom}, it is strictly dominant for the arm to report 0. The reward the arm receives in expectation for this step is $(u + \log(\bar{w}_{i}-M))\mu_iB$.

For step (2), assume that $i=i^*$; otherwise, arm $i$ is played $B$ times, and the dominant strategy is to report $0$ and receive expected reward $\mu_{i}B$. Depending on what happened in step (1), there are two cases; either $w' -M\leq \mu_{i}$, or $w' -M> \mu_{i}$. Similarly as the proof of Lemma \ref{lem:stra-dom}, we know that if $w' -M\leq \mu_{i}$, the arm should play $w'-M$ for the next $R$ rounds (not defecting) and report 0 for $B$ bonus rounds. If $w' -M> \mu_{i}$, the arm should play $0$ (defecting immediately).

For step (1), similar as the proof of Lemma \ref{lem:stra-dom}, the expected reward of arm $i$ is either

\begin{equation*}
G = (\mu_{i} -\bar{w}_{i})B + B\mu_{i} + \max(u + \log(\bar{w}_{i}-M), 0)B\mu_{i}
\end{equation*}
or
\begin{equation*}
G' = \min(0, R(\mu_{i}-w'+M)) + (\mu_{i} -\bar{w}_{i})B + B\mu_{i} + \max(u + \log(\bar{w}_{i}-M), 0)B\mu_{i}
\end{equation*}
Using the same argument as the proof of Lemma \ref{lem:stra-dom}, we know arm $i$'s dominant strategy is to make $\bar{w}_{i} = \mu_i + M$. 
\end{proof}

\begin{theorem}
\label{thm:stonstra}
If all the strategic arms use their dominant strategies in Lemma \ref{lem:nstra-dom}, then the principal will get at least $\max(u_s,u_n)T - o(T)$ with probability $1-o(1/T)$. Here $u_s$ is the second largest mean of the strategic arms and $u_n$ is the largest mean of the non-strategic arms. 
\end{theorem}

\begin{proof}
We prove that with high probability non-strategic arms' reported values don't deviate too much from their means. 

For each non-strategic arm $i$,  by Chernoff bound,
\[
\Pr[|\bar{w}_i - \mu_i| \geq M/ 2] \leq 2\exp(-(M/2)^2 B/2) \leq 1/(KT)^8
\]
By union bound, with probability $1-o(1/T)$, all non-strategic arm $i$ satisfy $|\bar{w}_i - \mu_i|\leq M/ 2$. From now on, we will assume we are in the case when $|\bar{w}_i - \mu_i|< M/ 2$, for all $i$ such that arm $i$ is a non-strategic arm. 

There are two cases:
\begin{enumerate}
\item Case 1: arm $i^*$ is a strategic arm. Then its easy to see that $w' \geq u_s+M$ and $w' \geq u_n - M/2$. And also $\mu_{i^*} = w_{i^*} - M \geq w' - M$. So only from the third step of Mechanism \ref{mech:sto-nstra-known-tacit}, the principal will get at least
\begin{eqnarray*}
&&(w' -M)R = \max(u_s, u_n - 3M/2) R \geq \max(u_s,u_n)R- 3MR/2 \\
&=& \max(u_s,u_n)T - \max(u_s,u_n)(u+3)BK -3MR/2 \\
&=&  \max(u_s,u_n)T - o(T).\\
\end{eqnarray*}
\item Case 2: arm $i^*$ is a non-strategic arm. We know that $\mu_{i^*} \geq w_{i^*} - M/2 \geq (w'-M) + M/2$. So by using Chernoff bound and union bound again, we know that arm $i^*$ will be stopped in the third with probability $o(1/T)$. We also know that $\mu_{i^*} \geq w_{i^*} - M/2 \geq u_s+M-M/2$ and $\mu_{i^*} \geq w_{i^*}-M/2 \geq u_n- M/2-M/2$. Using the same argument as Case 1, we know that only from the third step, the principal will get at least $ \max(u_s,u_n)T - o(T)$.
\end{enumerate}
\end{proof}

\paragraph*{Acknowledgements}

M.B. supported in part by an NSF CAREER award (CCF-1149888), NSF CCF-1215990, NSF CCF-1525342, NSF CCF-1412958, a Packard Fellowship in Science and Engineering, and the Simons Collaboration on Algorithms and Geometry.

Research completed in part while S.W. was a Research Fellow at the Simons Institute for the Theory of Computing.

\bibliographystyle{alpha}
\bibliography{bib}

\appendix

\section{Omitted Proofs}\label{sect:applb}

\begin{proof}[Proof of Theorem \ref{thm:alldists2arms}]
Let $D_1$ and $D_2$ be distributions with means $\mu_1$ and $\mu_2$ respectively, and both distributions supported on $[\sqrt{\delta/T}, 1]$.  We now describe the equilibrium strategy $S^*$ (the below description is for arm 1; $S^*$ for arm 2 is symmetric):
\begin{enumerate}
\item Set parameters $B = 6\sqrt{T\delta} /\mu_2$ and $\theta = \sqrt{\frac{\delta}{T}}$. 
\item Define $c_{1,t}$ to be the number times arm $1$ is pulled in rounds $1,...,t$. Similarly define $c_{2,t}$ to be the number times arm $2$ is pulled in rounds $1,...,t$.
\item For $t= 1,...,T$.
\begin{enumerate}
\item If there exists a $t' \leq t-1$ such that $c_{1,t'} /\mu_1 < c_{2,t'}/\mu_2 - B$, set $w_{1,t} = v_{1,t}$.
\item If the condition in (a) is not true, let $p_{1,t}$ be the probability that the principal will pick arm 1 in this round conditioned on the history (assuming player $2$ is also playing $S^*$), and let $p_{2,t} = 1-p_{1,t}$. Then:
\begin{enumerate}
\item If $c_{1,t-1}/\mu_1 < c_{2,t-1}/\mu_2$ and $p_{1,t}/\mu_1<p_{2,t}/\mu_2$, set $w_{1,t} = \theta $.
\item Otherwise, set $w_{1,t} = 0$.
\end{enumerate}
\end{enumerate}
\end{enumerate}

We will now show that $(S^*,S^*)$ is an $O(\sqrt{T\delta})$-Nash equilibrium. To do this, for any deviating strategy $S'$, we will both lower bound $u_1(M, S^*, S^*)$ and upper bound $u_1(M, S', S^*)$, hence bounding the net utility of deviation.

We begin by proving that $u_1(M, S^*, S^*) \geq \frac{\mu_1^2 T}{\mu_1+\mu_2} - O(\sqrt{T\delta})$. We need the following lemma.
\begin{lemma}
\label{lem:2eqapp}
If both arms are using strategy $S^*$, then with probability $\left(1-\frac{4}{T}\right)$, $|c_{1,t}/\mu_1 -c_{2,t}/\mu_2| \leq B$ for all $t\in[T]$. 
\end{lemma}

\begin{proof}
Assume that both arms are playing the strategy $S^*$ with the modification that they never defect (i.e. condition (a) in the above strategy is removed). This does not change the probability that $|c_{1,t} /\mu_1- c_{2,t}/\mu_2| \leq B$ for all $t \in [T]$. 

Define $R_{1,t} = \sum_{s=1}^{t}w_{1,s} - \sum_{s=1}^{t}w_{I_s, s}$ be the regret the principal experiences for not playing only arm 1. Define $R_{2,t}$ similarly. We will begin by showing that with high probability, these regrets are bounded both above and below. In particular, we will show that with probability at least $1-\frac{2}{T}$, $R_{i,t}$ lies in $ \left[-\frac{\mu_1}{\mu_2}(2\theta\sqrt{T\log T}+\delta), \delta\right]$ for all $t \in [T]$ and $i \in \{1, 2\}$.

To do this, note that there are two cases where the regrets $R_{1,t}$ and $R_{2,t}$ can possibly change. The first is when $p_{1,t}/\mu_1>p_{2,t}/\mu_2$ and $c_{1,t}/\mu_1>c_{2,t}/\mu_2$. In this case, the arms offer $(w_{1,t}, w_{2,t}) = (0, \theta)$. With probability $p_{1,t}$ the principal chooses arm $1$ and the regrets update to $(R_{1,t+1}, R_{2,t+1}) = (R_{1,t}, R_{2,t} + \theta)$, and with probability $p_{2,t}$ the principal chooses arm $2$ and the regrets update to $(R_{1,t+1}, R_{2, t+1}) = (R_{1, t}-\theta, R_{2,t})$. It follows that $\E[R_{1,t+1}/\mu_2+R_{2,t+1}/\mu_1|R_{1,t}/\mu_2+R_{2,t}/\mu_1] = R_{1,t}/\mu_2+R_{2,t}/\mu_1 + (p_{1,t}/\mu_1-p_{2,t}/\mu_2)\theta \geq R_{1,t}/\mu_2 + R_{2,t}/\mu_1$. 

In the second case, $p_{1,t}/\mu_1<p_{2,t}/\mu_2$ and $c_{2,t}/\mu_1<c_{1,t}/\mu_2$, and a similar calculation shows again that $\E[R_{1,t+1}/\mu_2+R_{2,t+1}/\mu_1|R_{1,t}/\mu_2+R_{2,t}/\mu_1] = R_{1,t}/\mu_2+R_{2,t}/\mu_1 + (p_{2,t}/\mu_2-p_{1,t}/\mu_1)\theta \geq R_{1,t} + R_{2,t}$. It follows that $R_{1,t}/\mu_2 + R_{2,t}/\mu_1$ forms a submartingale.

From the above analysis, it is also clear that $\left|(R_{1,t+1}/\mu_2 + R_{2,t+1}/\mu_1) - (R_{1,t}/\mu_2 + R_{2,t}/\mu_1)\right| \leq \theta/\mu_2$. It follows from Azuma's inequality that, for any fixed $t \in [T]$,

$$\Prob\left[R_{1,t} /\mu_2+ R_{2,t}/\mu_1 \leq -\frac{2\theta}{\mu_2}\sqrt{T\log T}\right] \leq \frac{1}{T^2}$$

Applying the union bound, with probability at least $1-\frac{1}{T}$, $R_{1,t} /\mu_2+ R_{2,t}/\mu_1 \geq -\frac{2\theta}{\mu_2}\sqrt{T\log T}$ for all $t\in[T]$. Furthermore, since the principal is using a $\left(T^{-2}, \delta\right)$-low-regret algorithm, it is also true that with probability at least $1-T^{-2}$ (for any fixed $t$) both $R_{1,t}$ and $R_{2,t}$ are at most $\delta$. Applying the union bound again, it is true that $R_{1,t} \leq \delta$ and $R_{2,t} \leq \delta$ for all $t$ with probability at least $1-\frac{1}{T}$. Finally, combining this with the earlier inequality (and applying union bound once more), with probability at least $1 - \frac{2}{T}$, $R_{i,t} \in \left[-\frac{\mu_1}{\mu_2}(2\theta\sqrt{T\log T}+\delta), \delta\right]$, as desired. For the remainder of the proof, condition on this being true.

We next proceed to bound the probability that (for a fixed $t$) $c_{1,t}/\mu_1 - c_{2,t}/\mu_2 \leq B$. Define the random variable $\tau-1$ to be the largest value $s \leq t$ such that $c_{1,\tau}/\mu_1 - c_{2, \tau}/\mu_2 \leq 0$ -- note that if $c_{1,t}/\mu_1 - c_{2,t}/\mu_2 \geq 0$, then $c_{1,s} /\mu_1- c_{2,s}/\mu_2 \geq 0$ for all $s$ in the range $[\tau, t]$. Additionally let $\Delta_{s}$ denote the $\pm 1$ random variable given by the difference $(c_{1,s} /\mu_1- c_{2,s}/\mu_2)-(c_{1,s-1}/\mu_1 - c_{2,s-1}/\mu_2)$. We can then write

\begin{eqnarray*}
c_{1,t} /\mu_1- c_{2,t} /\mu_2&\leq & \sum_{s = \tau+1}^{t} \Delta_{s} \\
&\leq & \sum_{s=\tau+1}^{t} \Delta_{s}\cdot \1_{p_{1,s}/\mu_1>p_{2,s}/\mu_2} + \sum_{s=\tau+1}^{t} \Delta_{s}\cdot \1_{p_{1,s}/\mu_1\leq p_{2,s}/\mu_2}
\end{eqnarray*}

Here the first summand corresponds to times $s$ where one of the arms offers $\theta$ (and hence the regrets change), and the second summand corresponds to times where both arms offer $0$. Note that since $c_{1,s}/\mu_1 \geq c_{2,s}/\mu_2$ in this interval, the regret $R_{2,s}$ increases by $\theta$ whenever $\Delta_{s} = 1/\mu_1$ (i.e., arm $1$ is chosen), and furthermore no choice of arm can decrease $R_{2,s}$ in this interval. Since we know that $R_{2,s}$ lies in the interval $\left[-\frac{\mu_1}{\mu_2}(2\theta\sqrt{T\log T}+\delta), \delta\right]$ for all $s$, this bounds the first sum by

$$\sum_{s=\tau+1}^{t} \Delta_{s}\cdot \1_{p_{1,s}>p_{2,s}} \leq \frac{\delta +\frac{\mu_1}{\mu_2}(2\theta\sqrt{T\log T}+\delta)}{\theta} \cdot (1/\mu_1) =\frac{1}{\mu_2}\left( \frac{2\delta}{\theta} + 2\sqrt{T\log T}\right)$$

On the other hand, when $p_{1,s}/\mu_1 \leq p_{2,s}/\mu_2$, then $\E[\Delta_{s}] = p_{1,s}/\mu_1-p_{2,s}/\mu_2\leq 0$. By Hoeffding's inequality, it then follows that with probability at least $1 - \frac{1}{T^2}$, 

$$\sum_{s=\tau+1}^{t} \Delta_{s}\cdot \1_{p_{1,s}\leq p_{2,s}} \leq \frac{2}{\mu_2}\sqrt{T\log T}$$

\noindent
Altogether, this shows that with probability at least $1 - \frac{1}{T^2}$,

$$c_{1,t} - c_{2,t} \leq \frac{1}{\mu_2}\left(\frac{2\delta}{\theta} + 4\sqrt{T\log T}\right) \leq 6\sqrt{T\delta} /\mu_2= B$$

The above inequality therefore holds for all $t$ with probability at least $1 - \frac{1}{T}$. Likewise, we can show that $c_{2,t} /\mu_2- c_{1,t} /\mu_1\leq B$ also holds for all $t$ with probability at least $1 - \frac{1}{T}$. Since we are conditioned on the regrets $R_{i,t}$ being bounded (which is true with probability at least $\frac{2}{T}$), it follows that $|c_{1,t}/\mu_1 - c_{2,t}/\mu_2| \leq B$ for all $t$ with probability at least $1 - \frac{4}{T}$. 

\end{proof}
By Lemma \ref{lem:2eq}, we know that with probability $1-\frac{4}{T}$, $|c_{1,t}/\mu_1 -c_{2,t}/\mu_2| \leq B$ throughout the mechanism. In this case, arm 1 never uses step (a), and $c_{1,T} \geq \frac{\mu_1}{\mu_1+\mu_2}T - \frac{\mu_1\mu_2}{\mu_1+\mu_2}B$. Therefore 

\begin{eqnarray*}
u_1(M, S^*, S^*) &\geq&  \left(1-\frac{4}{T}\right) \cdot (\mu_1-\theta) \cdot \left(\frac{\mu_1}{\mu_1+\mu_2}T - \frac{\mu_1\mu_2}{\mu_1+\mu_2}B\right)\\
&\geq & \frac{\mu_1^2T}{\mu_1+\mu_2} - O(\sqrt{T\delta})
\end{eqnarray*}

Now we will show that $u_1(M,S',S^*) \leq  \frac{\mu_1^2T}{\mu_1+\mu_2}  + O(\sqrt{T\delta})$. Without loss of generality, we can assume $S'$ is deterministic. Let $M_R$ be the deterministic mechanism when $M$'s randomness is fixed to some outcome $R$. Consider the situation when arm $1$ is using strategy $S'$, arm 2 is using strategy $S^*$ and the principal is using mechanism $M_R$. There are two cases:
\begin{enumerate}
\item $c_{1,t}/\mu_1 -c_{2,t}/\mu_2 \leq B$ is true for all $t\in[T]$. In this case, we have 

$$u_1(M_R,S',S^*) \leq c_{1,T} \cdot \mu_1 \leq \frac{\mu_1}{\mu_1+\mu_2}T + \frac{\mu_1\mu_2}{\mu_1+\mu_2}B.$$

\item There exists some $t$ such that $c_{1,t}/\mu_1 -c_{2,t}/\mu_2 > B$: Let $\tau_R+1$ be the smallest $t$ such that $c_{1,t}/\mu_1 -c_{2,t}/\mu_2 > B$. We know that $c_{1,\tau_R}/\mu_1 -c_{2,\tau_R} /\mu_2\leq B$. Therefore we have
\begin{eqnarray*}
u_1(M_R,S',S^*) &=&\sum_{t=1}^T (\mu_1 - w_{1,t}) \cdot \1_{I_t = 1} \\
&=& \sum_{t=1}^T (\mu_1 - w_{2,t}) \cdot \1_{I_t = 1} + \sum_{t=1}^T (w_{2,t} - w_{1,t}) \cdot \1_{I_t = 1} \\
&\leq& c_{1,\tau_R} \mu_1 + \mu_1 + (T-\tau_R-1) \max(\mu_1-\mu_2,0) + \sum_{t=1}^T (w_{2,t} - w_{1,t}) \cdot \1_{I_t = 1} \\
&\leq& \mu_1\left( \frac{\mu_1}{\mu_1+\mu_2}\tau_R + \frac{\mu_1\mu_2}{\mu_1+\mu_2}B\right) + \mu_1 + (T-\tau_R-1) \frac{q^2_1}{\mu_1+\mu_2} \\
&& +\sum_{t=1}^T (w_{2,t} - w_{1,t}) \cdot \1_{I_t = 1} \\
&\leq&  \frac{\mu_1^2}{\mu_1+\mu_2}T + \frac{\mu_1\mu_2}{\mu_1+\mu_2}B+\mu_1+ \sum_{t=1}^T (w_{2,t} - w_{1,t}) \cdot \1_{I_t = 1}. \\
\end{eqnarray*}
\end{enumerate}
In general, we thus have that
\[
u_1(M_R,S',S^*) \leq  \frac{\mu_1^2}{\mu_1+\mu_2}T + \frac{\mu_1\mu_2}{\mu_1+\mu_2}B+\mu_1 + \max\left(0, \sum_{t=1}^T (w_{2,t} - w_{1,t}) \cdot \1_{I_t = 1}\right). \\
\]
Therefore
\begin{eqnarray*}
u_1(M,S',S^*) &=& \E_R[u_1(M_R,S',S^*)] \\
&\leq&  \frac{\mu_1^2}{\mu_1+\mu_2}T + \frac{\mu_1\mu_2}{\mu_1+\mu_2}B+\mu_1 + \E_R\left[ \max\left(0, \sum_{t=1}^T (w_{2,t} - w_{1,t}) \cdot \1_{I_t = 1}\right)\right]. \\ 
\end{eqnarray*}
Notice that $\sum_{t=1}^T (w_{2,t} - w_{1,t}) \cdot \1_{I_t = 1}$ is the regret of not playing arm 2 (i.e., $R_2$ in the proof of Lemma \ref{lem:2eq}). Since the mechanism $M$ is $(\rho, \delta)$ low regret, with probability $1-\rho$, this sum is at most $\delta$ (and in the worst case, it is bounded above by $T \mu_2$). We therefore have that:

\begin{eqnarray*}
u_1(M,S',S^*) &\leq&  \frac{\mu_1^2}{\mu_1+\mu_2}T + \frac{\mu_1\mu_2}{\mu_1+\mu_2}B+\mu_1+ \delta + \rho T \mu_2 \\
&\leq &  \frac{\mu_1^2}{\mu_1+\mu_2}T + O(\sqrt{T\delta})
\end{eqnarray*}

From this and our earlier lower bound on $u_1(M, S^*, S^*)$, it follows that $u_1(M, S',S^*) - u_1(M, S^*, S^*) \leq O(\sqrt{T\delta})$, thus establishing that $(S^*, S^*)$ is an $O(\sqrt{T\delta})$-Nash equilibrium for the arms.

Finally, to bound the revenue of the principal, note that if the arms both play according to $S^*$ and $|c_{1,t}/\mu_1 - c_{2,t}/\mu_2| \leq B$ for all $t$ (so they do not defect), the principal gets a maximum of $T\theta = O(\sqrt{T\delta})$ revenue overall. Since (by Lemma \ref{lem:2eq}) this happens with probability at least $1 - \frac{4}{T}$ (and the total amount of revenue the principal is bounded above by $T$), it follows that the total expected revenue of the principal is at most $O(\sqrt{T\delta})$. 

\end{proof}

\begin{proof}[Proof of Theorem \ref{thm:advtacit}]
As in the proof of Theorem \ref{thm:advtacit2arms}, let $\mu_i$ denote the mean value of the $i$th arm's distribution $D_i$ (supported on $[\sqrt{K\delta/T}, 1]$). Without loss of generality, further assume that $\mu_{1} \geq \mu_{2} \geq \dots \geq \mu_{K}$. We will show that as long as $\mu_{1}-\mu_{2} \leq \frac{\mu_1}{K}$, there exists some $O(\sqrt{KT\delta})$-Nash equilibrium for the arms where the principal gets at most $O(\sqrt{KT\delta})$ revenue.

We begin by describing the equilibrium strategy $S^*$ for the arms. Let $c_{i,t}$ denote the number of times arm $i$ has been pulled up to time $t$. As before, set $B = 7\sqrt{KT\delta}$ and set $\theta = \sqrt{\frac{K\delta}{T}}$. The equilibrium strategy for arm $i$ at time $t$ is as follows:

\begin{enumerate}
\item
If at any time $s\leq t$ in the past, there exists an arm $j$ with $c_{j, s} - c_{i, s} \geq B$, defect and offer your full value $w_{i,t} = \mu_{i}$.
\item
Compute the probability $p_{i,t}$, the probability that the principal will pull arm $i$ conditioned on the history so far. 
\item
Offer $w_{i,t} = \theta(1-p_{i,t})$. 
\end{enumerate} 

We begin, as before, by showing that if all parties follow this strategy, then with high probability no one will ever defect. 

\begin{lemma}
\label{lem:keq}
If all arms are using strategy $S^*$, then with probability $\left(1-\frac{3}{T}\right)$, $|c_{i,t} -c_{j,t}| \leq B$ for all $t\in[T], i,j\in[K]$. 
\end{lemma}
\begin{proof}
As before, assume that all arms are playing the strategy $S^*$ with the modification that they never defect. This does not change the probability that $|c_{i,t} - c_{j,t}| \leq B$ for all $t\in[T], i,j\in[K]$. 

Define $R_{i,t} = \sum_{s=1}^{t}w_{i,s} - \sum_{s=1}^{t}w_{I_s, s}$ be the regret the principal experiences for not playing only arm $i$ up until time $t$. We begin by showing that with probability at least $1-\frac{2}{T}$, $R_{i,t}$ lies in $[-K\theta\sqrt{T\log T} - (K-1)\delta, \delta]$ for all $t \in [T]$ and $i \in[K]$.

To do this, first note that since the principal is using a $(T^{-2}, \delta)$-low-regret algorithm, with probability at least $1-T^{-2}$ the regrets $R_{i,t}$ are all upper bounded by $\delta$ at any fixed time $t$. Via the union bound, it follows that $R_{i,t} \leq \delta$ for all $i$ and $t$ with probability at least $1-\frac{1}{T}$.

To lower bound $R_{i,t}$, we will first show that $\sum_{i=1}^{K}R_{i,t}$ is a submartingale in $t$. Note that, with probability $p_{j,t}$, $R_{i, t+1}$ will equal $R_{i, t} + \theta((1-p_{j,t})-(1-p_{i,t}))$. We then have

\begin{eqnarray*}
\E\left[\sum_{i=1}^{K}R_{i,t+1}\middle|\sum_{i=1}^{K}R_{i, t}\right] &=& 
\sum_{i=1}^{K} R_{i,t} + \sum_{i=1}^{K}p_{i,t}\sum_{j=1}^{K}\theta((1-p_{j,t})-(1-p_{i,t})) \\
&=& \sum_{i=1}^{K} R_{i,t} + \sum_{i=1}^{K}p_{i,t}\sum_{j=1}^{K}\theta(p_{i,t} - p_{j,t}) \\
&=& \sum_{i=1}^{K} R_{i,t} + \theta\sum_{i=1}^{K}p_{i,t}(Kp_{i,t} - 1) \\
&=& \sum_{i=1}^{K} R_{i,t} + \theta\left(K\sum_{i=1}^{K}p_{i,t}^2 - \sum_{i=1}^{K}p_{i,t}\right) \\
&\geq& \sum_{i=1}^{K} R_{i,t}
\end{eqnarray*}

\noindent
where the last inequality follows by Cauchy-Schwartz. It follows that $\sum_{i=1}^{K}R_{i,t}$ forms a submartingale.

\sloppy
Moreover, note that (since $|p_{i}-p_{j}| \leq 1$) $|R_{i,t+1} - R_{i,t}| \leq \theta$. It follows that $\left|\sum_{i=1}^{K}R_{i,t+1} - \sum_{i=1}^{K}R_{i,t}\right| \leq K\theta$ and therefore by Azuma's inequality that, for any fixed $t \in [T]$,

$$\Pr\left[\sum_{i=1}^{K}R_{i,t} \leq -2K\theta\sqrt{T\log T}\right] \leq \frac{1}{T^2}.$$

With probability $1-\frac{1}{T}$, this holds for all $t \in [T]$. Since (with probability $1-\frac{1}{T}$) $R_{i, t} \leq \delta$, this implies that with probability $1-\frac{2}{T}$, $R_{i,t} \in \left[-2K\theta\sqrt{T\log T} - (K-1)\delta, \delta\right]$. 

We next proceed to bound the probability that $c_{i,t} - c_{j,t} > B$ for a $i$, $j$, and $t$. Define

$$S^{(i,j)}_{t} = \left(c_{i,t} - c_{j,t} + \frac{1}{\theta}(R_{i,t} - R_{j,t})\right).$$

We claim that $S^{(i,j)}_{t}$ is a martingale. To see this, we first claim that $R_{i,t+1} - R_{j,t+1} = R_{i, t} - R_{j,t} - \theta(p_{i,t} - p_{j,t})$. Note that, if arm $k$ is pulled, then $R_{i,t+1} = R_{i,t} + \theta((1-p_{i,t}) - (1-p_{k,t})) = R_{i,t} + \theta(p_{k,t} - p_{i,t})$ and similarly, $R_{j,t+1} = R_{j,t} + \theta(p_{k,t} - p_{j,t})$. It follows that $R_{i,t+1} - R_{j,t+1} = R_{i, t} - R_{j,t} - \theta(p_{i,t} - p_{j,t})$.

Secondly, note that (for any arm $k$) $\E[c_{k,t+1}-c_{k,t} | p_{t}] = p_{k,t}$, and thus $\E[c_{i,t+1} - c_{j, t+1} - (c_{i,t} - c_{j,t})| p_{t}] = p_{i,t} - p_{j,t}$. It follows that

\begin{eqnarray*}
\E[S^{(i,j)}_{t+1}-S^{(i,j)}_{t}|p_{t}] &=& \E[(c_{i,t+1} - c_{j,t+1})-(c_{i,t} -c_{j,t}) | p_{t}] \\
&&+ \frac{1}{\theta}\E[(R_{i,t+1}-R_{j,t+1})-(R_{i,t}-R_{j,t}) | p_{t}] \\
&=& (p_{i,t} - p_{j,t}) - (p_{i,t} - p_{j,t}) \\
&=& 0
\end{eqnarray*}

and thus that $\E[S^{(i,j)}_{t+1}|S^{(i,j)}_{t}] = S^{(i,j)}_{t}$, and thus that $S^{(i,j)}_{t}$ is a martingale. Finally, note that $|S^{(i,j)}_{t+1} - S^{(i,j)}_{t}| \leq 2$, so by Azuma's inequality

$$\mathrm{Pr}\left[S^{(i,j)}_{t} \geq 4\sqrt{T\log(TK)}\right] \leq (TK)^{-2}$$

\sloppy
Taking the union bound, we find that with probability at least $1- \frac{1}{T}$, $S^{(i,j)} \leq 4\sqrt{T\log(TK)}$ for all $i$, $j$, and $t$. Finally, since with probability at least $1-\frac{2}{T}$ each $R_{i,t}$ lies in $\left[-2K\theta\sqrt{T\log T} - (K-1)\delta, \delta\right]$, with probability at least $1-\frac{3}{T}$ we have that (for all $i$, $j$, and $t$)

\begin{eqnarray*}
c_{i,t} - c_{j,t} &=& S^{(i,j)}_{t} - \frac{1}{\theta}(R_{i,t} - R_{j,t}) \\
&\leq& 4\sqrt{T\log(TK)} + \frac{1}{\theta}\left|R_{i,t} - R_{j,t}\right| \\
&\leq & 4\sqrt{T\log(TK)} + 2K\sqrt{T\log T} + \frac{K\delta}{\theta} \\
&\leq & \frac{7K\delta}{\theta} \\
&=& 7K\sqrt{T\delta} \\
&=& B
\end{eqnarray*}
\end{proof}

By Lemma \ref{lem:keq}, we know that with probability $1-\frac{3}{T}$, $|c_{i,t} -c_{j,t}| \leq B$ for all $t\in[T], i,j\in[K]$. In this case, arm 1 never defect, and $c_{1,T} \geq T/K-B$. Therefore 

\begin{eqnarray*}
u_1(M, S^*, S^*) &\geq&  \left(1-\frac{3}{T}\right) \cdot (\mu_1-\theta) \cdot (T/K-B) \\
&\geq & \frac{\mu_1T}{K}\left(1 - \frac{3}{T} - \frac{\theta}{\mu_1} - \frac{BK}{T}\right) \\
&=& \frac{\mu_1T}{K} - 3\mu_1/K - \frac{\theta T}{K} - B\mu_1 \\
&\geq & \frac{\mu_1T}{K} - O(\sqrt{KT\delta})
\end{eqnarray*}

Now we are going to show that $u_1(M,S',S^*) \leq \frac{\mu_1T}{K} + O(\sqrt{KT\delta})$. Without loss of generality, we can assume $S'$ is deterministic. Let $M_R$ be the deterministic mechanism when $M$'s randomness is fixed to some outcome $R$. Consider the situation when arm $1$ is using strategy $S'$, arm 2 is using strategy $S^*$ and the principal is using mechanism $M_R$. There are two cases:
\begin{enumerate}
\item $c_{i,t} -c_{j,t} \leq B$ is true for all $t\in[T]$ and $i,j \in [K]$. In this case, we have 
\[
u_1(M_R,S',S^*) \leq c_{1,T} \cdot \mu_1 \leq \mu_1(T+(K-1)B)/K.
\]
\item There exists some $t\in[T]$ and $i,j\in[K]$ such that $c_{i,t} -c_{j,t} > B$: Let $\tau_R+1$ be the smallest $t$ such that $c_{i,t} -c_{j,t} > B$ for some $i,j\in[K]$. We know that $c_{1,\tau_R} -c_{i,\tau_R} \leq B$ for all $i \in [K]$. Therefore we have
\begin{eqnarray*}
u_1(M_R,S',S^*) &=&\sum_{t=1}^T (\mu_1 - w_{1,t}) \cdot \1_{I_t = 1} \\
&=& \sum_{t=1}^T (\mu_1 - w_{2,t}) \cdot \1_{I_t = 1} + \sum_{t=1}^T (w_{2,t} - w_{1,t}) \cdot \1_{I_t = 1} \\
&\leq& c_{1,\tau_R} \mu_1 + \mu_1 + (T-\tau_R-1) \max(\mu_1-\mu_2,0) + \sum_{t=1}^T (w_{2,t} - w_{1,t}) \cdot \1_{I_t = 1} \\
&\leq& \mu_1(\tau_R+B)/K + \mu_1 + (T-\tau_R-1) (\mu_1/K) +\sum_{t=1}^T (w_{2,t} - w_{1,t}) \cdot \1_{I_t = 1} \\
&\leq& \mu_1T/K + \mu_1(B+1)(K-1)/K + \sum_{t=1}^T (w_{2,t} - w_{1,t}) \cdot \1_{I_t = 1}. \\
\end{eqnarray*}
In $M_R$, we also have
\begin{eqnarray*}
 \sum_{t=1}^T (w_{2,t} - w_{1,t}) \cdot \1_{I_t = 1} &=&  \sum_{t=1}^T (w_{2,t} - w_{I_t,t} ) - \sum_{t=1}^T (w_{2,t} - w_{I_t,t} )\cdot \1_{I_t \neq 1} \\
 &\leq& \sum_{t=1}^T (w_{2,t} - w_{I_t,t} )  + \sum_{t=1}^{\tau_R} w_{I_t,t} \cdot \1_{I_t \neq 1} - \sum_{t=\tau_R+1}^T (\mu_2-\mu_{I_t}) \cdot \1_{I_t \neq 1} \\
 &\leq& \sum_{t=1}^T (w_{2,t} - w_{I_t,t} ) + T(\theta+B/T) + 0.
 \end{eqnarray*}
\end{enumerate}
In general, we thus have that
\[
u_1(M_R,S',S^*) \leq \mu_1T/K + \mu_1(B+1)(K-1)/K  + \max\left(0, \sum_{t=1}^T (w_{2,t} - w_{I_t,t} ) + T\theta+B\right). \\
\]
Therefore
\begin{eqnarray*}
u_1(M,S',S^*) &=& \E_R[u_1(M_R,S',S^*)] \\
&\leq& \mu_1T/K + \mu_1(B+1)(K-1)/K + \E_R\left[ \max\left(0, \sum_{t=1}^T (w_{2,t} - w_{I_t,t} ) + T\theta+B\right)\right]. \\ 
\end{eqnarray*}
Notice that $ \sum_{t=1}^T (w_{2,t} - w_{I_t,t} )$ is the regret of not playing arm 2. Since the mechanism $M$ is $(\rho, \delta)$ low regret, with probability $1-\rho$, this sum is at most $\delta$ (and in the worst case, it is bounded above by $T \mu_2$). We therefore have that:

\begin{eqnarray*}
u_1(M,S',S^*) &\leq& \mu_1T/K + \mu_1(B+1)(K-1)/K  + \delta + \rho T \mu_ + T\theta + B \\
&\leq & \frac{\mu_1T}{K} + O(\sqrt{KT\delta}).
\end{eqnarray*}

From this and our earlier lower bound on $u_1(M, S^*, S^*)$, it follows that $u_1(M, S',S^*) - u_1(M, S^*, S^*) \leq O(\sqrt{KT\delta})$, thus establishing that $(S^*, S^*)$ is an $O(\sqrt{KT\delta})$-Nash equilibrium for the arms.

Finally, to bound the revenue of the principal, note that if the arms both play according to $S^*$ and $|c_{i,t} - c_{j,t}| \leq B$ for all $t\in[T],i,j\in[K]$ (so they do not defect), the principal gets a maximum of $T\theta = O(\sqrt{KT\delta})$ revenue overall. Since (by Lemma \ref{lem:2eq}) this happens with probability at least $1 - \frac{3}{T}$ (and the total amount of revenue the principal is bounded above by $T$), it follows that the total expected revenue of the principal is at most $O(\sqrt{KT\delta})$. 
\end{proof}

\end{document}